\documentclass[letterpaper, 10 pt, conference]{ieeeconf} 

\IEEEoverridecommandlockouts

\makeatletter

\let\proof\@undefined
\let\endproof\@undefined
\makeatother

\usepackage{graphics}
\usepackage{epsfig}
\usepackage{mathptmx}
\usepackage{amsmath}
\usepackage{amssymb}
\usepackage{amsthm}
\usepackage{optidef}
\usepackage{tikz}
\usepackage[english]{babel}
\usepackage[utf8]{inputenc}
\usepackage{caption}
\usepackage{subcaption}
\usepackage{lipsum}
\usepackage[mathscr]{euscript}
\usepackage{thmtools}
\usepackage{caption}
\usepackage{subcaption}
\usepackage{stmaryrd}

\usepackage{color}

\DeclareMathAlphabet{\mathcal}{OMS}{cmsy}{m}{n}
\newcommand*{\QEDB}{\hfill\ensuremath{\square}}

\title{\LARGE \bf
Control of Multi-Agent Systems with Finite Time Control Barrier Certificates and Temporal Logic
}

\author{Mohit Srinivasan$^{1}$, Samuel Coogan$^{2}$, and Magnus Egerstedt$^{3}$
\thanks{*This work was supported by DARPA under the Grant N66001-17-2-4059.}
\thanks{$^{1}$Mohit Srinivasan and $^{3}$Magnus Egerstedt are with the School of Electrical and Computer Engineering, Georgia Institute of Technology,
	Atlanta, Georgia 30332, USA
        {\tt\small mohit.srinivasan@gatech.edu; magnus@gatech.edu}}%
\thanks{$^{2}$Samuel Coogan is with the School of Electrical and Computer Engineering and the School of Civil and Environmental Engineering, Georgia
	Institute of Technology, Atlanta, Georgia 30332, USA
        {\tt\small sam.coogan@gatech.edu}}%
}

\newtheorem{prop}{Proposition}
\newtheorem{defn}{Definition}
\newcommand{\Gglobe}{\mathcal{G}^\text{globe}}
\newcommand{\hglobe}{h^\text{globe}}
\newcommand{\piglobe}{\pi^\text{globe}}

\begin{document}

\captionsetup[figure]{labelfont=bf}
\captionsetup[subfigure]{labelfont=bf}

\maketitle
\thispagestyle{empty}
\pagestyle{empty}

\begin{abstract}
In this paper, a method to synthesize controllers using finite time convergence control barrier functions guided by linear temporal logic specifications for continuous time multi-agent dynamical systems is proposed. Finite time convergence to a desired set in the state space is guaranteed under the existence of a suitable finite time convergence control barrier function. In addition, these barrier functions also guarantee forward invariance once the system converges to the desired set. This allows us to formulate a theoretical framework which synthesizes controllers for the multi-agent system. These properties also enable us to solve the reachability problem in continuous time by formulating a theorem on the composition of multiple finite time convergence control barrier functions. This approach is more flexible than existing methods and also allows for a greater set of feasible control laws. Linear temporal logic is used to specify complex task specifications that need to be satisfied by the multi-agent system. With this solution methodology, a control law is synthesized that satisfies the given temporal logic task specification. Robotic experiments are provided which were performed on the Robotarium multi-robot testbed at Georgia Tech.

\end{abstract}

\section{INTRODUCTION}
Complex mission specifications require provably correct controllers that satisfy the task specification infinitely often. To that end, we address the issue of synthesizing a control architecture for multi-agent systems, subject to LTL specifications. We propose a control architecture which uses finite time convergence control barrier functions (hereafter to be known as \emph{finite time barrier certificates}) and temporal logic to solve the continuous time reachability problem in multi-agent dynamical systems. In particular, we use finite time barrier certificates, introduced in \cite{c15} in the context of composition of different behaviors for multi-robot systems, to guarantee finite time reachability to desired regions in the state space and linear temporal logic for specifying complex task specifications to be satisfied by the system.

The contributions of this paper are threefold. First, we introduce the notion of composition of multiple finite time barrier certificates using addition which provides feasible solutions in cases where using methods such as the one followed in \cite{c15} can lead to infeasibility. Our framework results in a larger set of feasible control laws as compared to methods such as the one followed in \cite{c15}. This allows for more flexibility when the task specification is more complex. Second, we use finite time control barrier certificates inspired by \cite{c15} for a continuous time multi-agent system in the context of motion planning. Third, we use a discretization free approach inspired by \cite{c1} which allows us to leverage key ideas from automata theory and temporal logic, in conjunction with finite time barrier certificates.

Control of multi-agent systems has been studied extensively over the past few years in a plethora of context and settings. We cite some of the recent work in this domain which are most related to the concepts provided in this paper. Papers \cite{c4}, \cite{c5}, discuss the use of control barrier functions (described in Section II) for collision avoidance in multi-robot systems, while \cite{c6}, \cite{c7} apply these principles to adaptive cruise control and automotive systems. In \cite{c2}, \cite{c9}, \cite{c10}, the authors discuss a verification method for non-linear systems as well as for continuous and hybrid systems in stochastic and worst case settings. In all these papers, the primary focus is safety whereas in our paper, we shift the focus towards finite time reachability for continuous time dynamical systems. LTL based motion planning has also been the subject of recent study \cite{c17, c18, c19}. Papers \cite{c8}, \cite{c11}, \cite{c12}, \cite{c13}, \cite{c14}, synthesize robust optimal controllers for multi-robot systems subject to LTL as well as other temporal logic languages. In the aforementioned papers, a discretization of the system dynamics is involved which we avoid in our paper.

This paper is organized as follows. Section II presents mathematical tools required for our solution approach. In Section III, we formulate the problem statement tackled in this paper. Section IV combines ideas from temporal logic and barrier certificates. In Section V, we formulate our theorem on composition of multiple finite time barrier certificates, and also discuss some important remarks regarding the same. Section VI provides an illustrative example that highlights our solution methodology. Section VII discusses the simulation and experimental results conducted on the Robotarium testbed facility at Georgia Tech \cite{c16}, respectively. Section VIII provides concluding remarks.

\section{MATHEMATICAL PRELIMINARIES}
This section discusses finite time barrier certificates and linear temporal logic.
\subsection{Control Barrier Functions}

In this paper, we use finite time convergence control barrier functions introduced in \cite{c15} as a building block for our control architecture. Consider a control affine dynamical system
\begin{equation}
\label{ControlAffineSystem}
\dot x = f(x) + g(x)u \text{,}
\end{equation}
where $f(x)$ and $g(x)$ are locally Lipschitz continuous, $x \in \mathcal{X}\subseteq \mathbb{R}^{n}$, and $u \in \mathbb{R}^{m}$. 

\begin{defn}{\cite{c15}}
A function $h(x):\mathcal{X}\to\mathbb{R}$ is a \textit{finite time convergence control barrier function} if there exists real parameters $\rho \in [0, 1)$ and $\gamma > 0$ such that for all $x \in \mathcal{X}$,
\begin{equation}
\sup_{u\in \mathbb{R}^{m}} \left\{ L_{f}h(x) + L_{g}h(x)u + \gamma \cdot \text{sign}(h(x)) \cdot |h(x)|^{\rho} \right\} \geq 0 \text{.}
\end{equation}
\end{defn} \QEDB

Finite time convergence control barrier functions are used to establish finite time reachability. Thus, if $h(x)$ is a finite time convergence control barrier function, then there exists a control input $u$ that drives the state of the system $x$ to the set $\{x \in \mathbb{R}^{n}| h(x)\geq 0\}$ in finite time, as formalized next.

\begin{prop}{\cite{c15}}
Let $h(x)$ be a finite time convergence control barrier function for \eqref{ControlAffineSystem}, and for all $x\in\mathcal{X}$, define 
\begin{multline}
\label{eq:Ux}
\underline{\mathcal{U}}(x) = \bigg\{ u \in \mathbb{R}^{m} \bigg| L_{f}h(x) + L_{g}h(x)u + \\
\gamma \cdot sign(h(x)) \cdot |h(x)|^{\rho} \geq 0 \bigg\} \text{.}
\end{multline}
Then, for any initial condition $x_{0} \in \mathcal{X}$ and any continuous feedback control $u : \mathcal{X} \rightarrow \mathbb{R}^{m}$ satisfying $u(x) \in \underline{\mathcal{U}}(x)$ for all $x \in \mathcal{X}$, the system will be driven to the set $\mathcal{G}:=\{x \in \mathbb{R}^{n}| h(x)\geq 0\}$ in a finite time $0 < T < \infty$ such that $x(T)\in \mathcal{G}$, where the time bound is given by $T = \frac{|h(x_0)|^{1 - \rho}}{\gamma (1 - \rho)}$ \cite{c20}, and renders the set $\mathcal{G}$ forward invariant.
\end{prop} \QEDB

Above, $L_{f}h(x) = \frac{\partial h(x)}{\partial x}f(x)$ and $L_{g}h(x) = \frac{\partial h(x)}{\partial x}g(x)$ are the Lie derivatives of $h(x)$ along $f(x)$ and $g(x)$ respectively.

Given a finite time convergence control barrier function $h(x)$, a closed formed expression for $\underline{\mathcal{U}}(x)$ as in \eqref{eq:Ux} is rarely available. However, in practice, this formulation is amenable to efficient online computation of feasible control inputs. In particular, for fixed $x$, the requirement that $u \in \underline{\mathcal{U}}(x)$ becomes a linear constraint and we define a minimum energy quadratic program (QP) as 
\begin{equation}
\begin{aligned}
& \underset{u \in \mathbb{R}^{m}}{\text{min}}
\quad ||u||_{2}^{2}\\
& \text{s.t \quad} u \in \underline{\mathcal{U}}(x) \text{.}
\end{aligned}
\end{equation}

This QP is solved with the finite time barrier certificate as the constraint on the control law $u(x)$ and returns the minimum energy control law that drives the system to the goal set $\mathcal{G}=\{x \in \mathbb{R}^{n}| h(x)\geq 0\}$ in finite time. We will reference this idea of a QP based controller throughout this paper in the context of our theorem and analysis.

\subsection{Linear Temporal Logic (LTL)}
Complex and rich system properties can be expressed succinctly using LTL, and one of the main advantages of LTL is the ease with which high level system objectives can be formalized. LTL formulas are developed using atomic propositions which label regions of interest within the state space. These formulas are built using a specific grammar. LTL formulas without the next operator are given by the following grammar \cite{c3}:
\begin{equation}
\phi = \pi | \neg \phi | \phi \vee \phi | \phi \mathcal{U} \phi
\end{equation}
where $\pi$ is a member of the set of atomic propositions, and $\phi$ represents a LTL specification. From the negation ($\neg$) and the disjunction ($\vee$) operators, we can define the conjunction ($\wedge$), implication ($\rightarrow$), and equivalence ($\leftrightarrow$) operators. We can thus derive for example, the eventually ($\Diamond$) and always ($\Box$) operators as $\Diamond \phi = \top \mathcal{U} \phi$ and $\Box \phi = \neg \Diamond \neg \phi$ respectively.

\section{PROBLEM FORMULATION}
Consider a multi-agent system consisting of $N$ robots with index set $\mathcal{I} = \left\{ 1, 2, \dots, N\right\}$. The  dynamics  for each agent $i \in \mathcal{I}$ is 
\begin{equation}
\dot x_{i} = u_{i}
\end{equation}
where $x_{i} \in \mathbb{R}^{n}$ and $u_{i} \in \mathbb{R}^{n}$ for some $n\in\mathbb{N}$.

Let $\mathcal{D}  \subset \mathbb{R}^{n}$ be the closed and connected domain for the agents and suppose $\mathcal{D}$ can be written as a super zero level set of a function $h_{\mathcal{D}}$, that is, $\mathcal{D} = \left\{ x \in \mathbb{R}^{n} | h_{\mathcal{D}}(x) \geq 0 \right\}$. The state space for the multi-agent system is then $\mathcal{D}^{N} \subset \mathbb{R}^{Nn}$. Consider a finite set $\mathcal{R}$ of regions of interest in the domain such that for all $r \in \mathcal{R}$, $r \subset \mathcal{D}$ and there exists a continuously differentiable function $h_{r} : \mathcal{D} \rightarrow \mathbb{R}$ such that $r = \left\{ x \in \mathcal{D} | h_{r}(x_i) \geq 0 \right\}$ for all $i \in \mathcal{I}$. Regions of interest may denote, for example, \emph{goal regions} that must be reached by an agent or agents, or it may denote regions that should be \emph{avoided} by an agent or agents. For each $r\in\mathcal{R}$ and $i\in\mathcal{I}$, let
\begin{equation}
\label{Goal}
\mathcal{G}_{i}^{r} = \left\{ x \in \mathcal{D}^{N} | h_{r}(x_{i}) \geq 0 \right\} \text{,}
\end{equation}
that is, $\mathcal{G}_{i}^{r}$ is the set of states for the multi-agent system for which agent $i$ is in region $r$.
Similarly, let
\begin{equation}
\label{Obstacle}
\overline{\mathcal{G}_{i}^{r}} = \left\{ x \in \mathcal{D}^{N} | h_{r}(x_{i}) < 0 \right\} \text{,}
\end{equation}
that is, $\overline{\mathcal{G}_{i}^{r}}$ is the set of states for the multi-agent system for which agent $i$ is outside the region $r$.

In addition to these regions of interest, we assume there exists a set of global conditions $\Gglobe_1, \Gglobe_2, \ldots,\Gglobe_C$ defined over the multiagent domain $\mathcal{D}^N$ such that for each $c\in \mathcal{C}:=\{1,2,\ldots,C\}$, $\Gglobe_c\subseteq \mathcal{D}^N$ and there exists a continuously differentiable function $\hglobe_c : \mathcal{D}^N\to\mathbb{R}$ such that $\Gglobe_c=\{x\in\mathcal{D}^N | \hglobe_c(x)\geq 0\}$. 

Similarly, we can write $\overline{\Gglobe_c}=\{x\in\mathcal{D}^N | \hglobe_c(x) < 0\}$. Such global conditions may include, for example, a connectivity constraint.

For each $\mathcal{G}_{i}^{r}$ with $i \in \mathcal{I}$, $r \in \mathcal{R}$, let 
\begin{equation}
\label{pie}
\pi_{i}^{r} = \left\{
        \begin{array}{ll}
            1 & \quad x\in\mathcal{G}^r_i\\
            0 & \quad otherwise.
        \end{array}
    \right.
\end{equation}
This means $\pi_{i}^{r} = 1$ if and only if agent $i$ is in region $r$. 

Similarly, for all $c\in\mathcal{C}$, let $\piglobe_c=1$ if and only if $x\in \Gglobe_c$.

The collection 
\begin{equation}
\label{AP}
 \Pi=\{\pi_{i}^{r} | i \in \mathcal{I}, r \in \mathcal{R}\}\cup \{\piglobe_c | c\in\mathcal{C}\}   
\end{equation}
constitutes the set of atomic propositions for the multi-agent system. For $\pi\in\Pi$, we will sometimes write $\mathcal{G}_\pi$ to denote the set that induces $\pi$, \emph{i.e.}, $\mathcal{G}_\pi=\mathcal{G}^r_i$ if $\pi=\pi^r_i$ for some $i\in\mathcal{I}$, $r\in \mathcal{R}$, or $\mathcal{G}_\pi=\Gglobe_c$ if $\pi=\piglobe_c$ for some $c\in\mathcal{C}$. Finally, for $a\subset \Pi$, we denote
\begin{equation}
 \llbracket a \rrbracket =\bigcap_{\pi\in a}\mathcal{G}_\pi\cap \bigcap_{\pi \in\Pi\backslash a}\overline{\mathcal{G}_{\pi}},
\end{equation}
that is, $x\in\llbracket a \rrbracket$ if and only if, for all $\pi\in\Pi$, $x\in \mathcal{G}_\pi$ if and only if $\pi \in a$.

Given a trajectory $x(t)$ of the multi-agent system, intuitively, the {trace} of the trajectory is the sequence of sets of atomic propositions that are satisfied along the trajectory. The following definition formally defines the trace of a trajectory of a system \cite{c2}.
\begin{defn}
An infinite sequence $\sigma = a_{0}a_{1}\dots$ where $a_{i} \subseteq {\Pi}$ for all $i \in \mathbb{N}$ is the {\normalfont trace of a trajectory} $x(t)$ if there exists an associated sequence $t_{0}t_{1}t_{2}\dots$ of time instances such that $t_{0} = 0$, $t_{k} \rightarrow \infty$ as $k \rightarrow \infty$ and for each $m \in \mathbb{N}$, $t_{m} \in \mathbb{R}_{\geq 0}$ satisfies the following conditions,
\begin{itemize}
\item $t_{m} < t_{m+1}$
\item $x(t_{m}) \in \llbracket a_m \rrbracket$,
\item If $a_{m} \neq a_{m+1}$, then for some $t_{m}^{'} \in [t_{m}, t_{m+1}]$, $x(t) \in \llbracket a_m \rrbracket$ for all $t \in (t_{m}, t_{m}')$, $x(t) \in \llbracket a_{m+1} \rrbracket$ for all $t \in (t_{m}', t_{m+1})$, and either $x(t_{m}') \in \llbracket a_m \rrbracket$ or $x(t_{m}') \in \llbracket a_{m+1} \rrbracket$.
\item If $a_m=a_{m+1}$ for some $m$, then $a_m=a_{m+k}$ for all $k>0$ and $x(t)\in \mathcal{G}_\pi$ (resp., $x(t)\not \in \mathcal{G}_{\pi}$) if $\pi\in a_m$ (resp., $\pi\not\in a_m$) for all $t\geq t_m$ for all $\pi\in\Pi$. \QEDB
\end{itemize}
\end{defn}
The last condition of the above definition implies that a trace contains a repeated set of atomic propositions only if this set holds for all future time, capturing, \emph{e.g.}, a stability condition of the multi-agent system. By forbidding repetitions otherwise, we ensure that each trajectory possesses a unique trace. This exclusion is without loss of generality since we only considered LTL specifications without the next operator. To that end, we define the problem we aim to solve in this paper.

\newtheorem*{obj*}{\textbf{System Objective}}
\begin{obj*}
Given a multi-agent system with initial condition $x(0)\in\mathcal{D}^N$ and a LTL specification $\phi$ over the set of atomic propositions $\Pi$, synthesize a control law such that the resulting trace of the system satisfies the specification $\phi$.
\end{obj*}

\section{LASSO-TYPE CONSTRAINED REACHABILITY OBJECTIVES}
To solve the above objective, we propose to use finite time barrier certificates to solve for a satisfying controller online. To this end, we note that it is well-known that if there exists a trace (that is, a sequence of sets of atomic propositions) that satisfies a given LTL specification, then there exists a trace satisfying the specification in \emph{lasso} or \emph{prefix-suffix} form \cite{c3}, where a trace $\sigma$ in lasso form consists of a prefix $\sigma_{\text{pre}}$ and suffix $\sigma_{\text{suff}}$ that are both finite sequences of sets of atomic propositions such that the trace $\sigma$ is equal to the prefix sequence followed by the suffix sequence repeated infinitely often. Such a lasso trace is denoted as $\sigma= \sigma_{\text{pre}}(\sigma_{\text{suff}})^\omega$ where $\omega$ signifies infinite repetition.

Because atomic propositions of the multi-agent system are defined as subsets of the domain, it is possible to interpret such lasso traces as sequences of constrained reachability problems in lasso form, which leads to our control synthesis methodology described in Section VI. To that end, we have the following definitions.
\begin{defn}
\label{def:reach}
Given two sets $\Sigma\subseteq \mathcal{D}^{N}$ and $\Gamma\subseteq \mathcal{D}^N$, the \emph{constrained reachability problem} $R(\Sigma,\Gamma)$ consists in finding a feedback control strategy $u:\Sigma \to \mathbb{R}^{Nn}$ for the multi-agent system such that for any $x(0)\in \Sigma$, there exists a finite time $0 < T < \infty$ satisfying $x(t)\in \Sigma$ for all $t\in[0,T]$ and $x(T)\in \Gamma$, where $x(t)$ is the trajectory of the system initialized at $x(0)$ subject to the control strategy $u(x)$.
\end{defn} \QEDB

With Definition 3, we formalize the definition of a constrained reachability problem induced by sets of atomic propositions.
\begin{defn}
Let $a_{1}^{\top}\subseteq \Pi$, $a_{1}^{\bot}\subseteq \Pi$, $a_{2}^{\top} \subseteq \Pi$, and $a_{2}^{\bot} \subseteq \Pi$ be sets of atomic propositions for the multi agent system. Let
\begin{align}
\label{reachprob}
    \Gamma&= \bigg( \bigcap_{\pi\in a_{2}^{\top}\backslash a_{1}^{\top}} \mathcal{G}_\pi \bigg) \cap \bigg( \bigcap\limits_{\pi \in a_{2}^{\bot} \backslash a_{1}^{\bot}} \overline{\mathcal{G}_{\pi}}\bigg) \\
\label{avoidprob}
    \Sigma&= \bigg( \bigcap_{\pi \in a_{1}^{\top} \cap a_{2}^{\top}} \mathcal{G}_\pi \bigg) \cap \bigg( \bigcap_{\pi \in a_{1}^{\bot} \cap a_{2}^{\bot}}\overline{\mathcal{G}_\pi} \bigg).
\end{align}
The reachability problem $R(\Sigma,\Gamma)$ is the {\normalfont Constrained Reachability problem} induced by the sets $a_{1}^{\top}$, $a_{1}^{\bot}$, $a_{2}^{\top}$, and $a_{2}^{\bot}$.
\end{defn} \QEDB 

Here, $a_{1}^{\top}$ represents the set of atomic propositions which are true before the reachability objective is executed, $a_{1}^{\bot}$ represents the set of atomic propositions which are false before the reachability objective is executed, $a_{2}^{\top}$ is the set of atomic propositions that must be true after the execution of the reachability objective and $a_{2}^{\bot}$ is the set of atomic propositions that must be false at the end of the execution of the reachability objective.

In the above definitions, $\Gamma$ represents the reachability set and $\Sigma$ represents the safety set. If we can solve the constrained reachability problem defined in Definition 4 by synthesizing a control law $u : \Sigma \rightarrow \mathbb{R}^{Nn}$ such that the conditions in Definition 3 are satisfied, then the multi-agent system will converge to the set $\Gamma$ in a finite time while remaining in the safety set $\Sigma$. We solve a series of constrained reachability problems which results in a system trajectory whose trace satisfies the given LTL specification.

\color{black}
\begin{defn}
A \emph{lasso-type constrained reachability sequence} is a sequence of constrained reachability problems in lasso form such that each subsequent safety set is compatible with the prior goal set. That is, a lasso-type constrained reachability sequence has the form
  \begin{align}
    \label{reachlasso}
\mathcal{R}_{lasso} = \bigg(R_1R_2\ldots R_k\bigg) \bigg(R_{k+1},R_{k+2}\ldots R_{k+\ell}\bigg)^\omega \text{,}
  \end{align}
where $k\geq 0$, $\ell\geq 1$, and each $R_j=R(\Sigma_j, \Gamma_j)$ for some $\Gamma_j,\Sigma_j\subset D^N$ satisfying $\Gamma_j\subseteq \Sigma_{j+1}$ for all $j\in\{1,2,\ldots,k+\ell\}$ and also $\Gamma_{k+\ell}\subseteq \Sigma_{k+1}$. The sequence $(R_1R_2\ldots R_k)$ is a finite horizon prefix objective and $(R_{k+1},R_{k+2}\ldots R_{k+\ell})$ is a finite suffix objective that is repeated infinitely often.
\end{defn} \QEDB

By the preceding discussion, if there exists a trace that satisfies a given LTL specification, then there exists a lasso-type constrained reachability sequence that, if feasible, guarantees that the multi-agent system satisfies  the LTL specification. Algorithms exist for automatically extracting traces in lasso form from a so-called \emph{B\"{u}chi Automaton} constructed from a LTL specification \cite{c3}. Choosing a good lasso sequence candidate from a list of possible lasso sequences is outside the scope of this paper, but is the subject of our current research work.

\section{COMPOSITE FINITE TIME CONTROL BARRIER CERTIFICATES}
In this section, we formulate a theorem on composition of multiple finite time barrier certificates for a general control affine system of the form \eqref{ControlAffineSystem}. This result is applicable to any system with control affine dynamics as in \eqref{ControlAffineSystem}.
\newtheorem{theorem}{\textbf{Theorem}}
\renewcommand{\qedsymbol}{$\blacksquare$}
\begin{theorem}
Consider a dynamical system in control affine form as in \eqref{ControlAffineSystem}. Given $\Gamma \subset \mathbb{R}^{n}$ defined by a collection of $q\geq 1$ functions $\left\{ h_{i}(x) \right\}_{i = 1}^{q}$ such that $\Gamma = \bigcap\limits_{i = 1}^{q} \left\{ x \in \mathbb{R}^{n} | h_{i}(x) \geq 0 \right\}$ and for $ i = \left\{ 1, 2, 3, . . . , q' \right\}$ with $q' < q$, $h_{i}(x)$ is bounded i.e. $h_i(x) < M_i$ for all $x \in \mathcal{D}$, for $M_i > 0$.\footnote{If all the functions are bounded, then $q' = q$ and so we will have only \eqref{Constraint1} as a constraint in the QP $\forall i \in \left\{ 1,2,\dots, q \right\}$} If there exists a collection $\left\{ \alpha_{i} \right\}_{i = 1}^{q'}$ with $\alpha_{i} \in \mathbb{R}_{> 0}$, parameters $\gamma > 0$, $\rho \in [0,1)$ and a continuous controller $u(x)$ where $u : \mathcal{D} \rightarrow \mathbb{R}^{m}$, such that for all $x \in \mathcal{D}$
\begin{multline}
\label{Constraint1}
\sum\limits_{i = 1}^{q'} \bigg\{ \alpha_{i}(L_{f}h_{i}(x) + L_{g}h_{i}(x)u(x)) \bigg\} + \\
\gamma \cdot sign\bigg(min \bigg\{ h_{1}(x), h_{2}(x),\dots, h_{q'}(x) \bigg\}\bigg) \geq 0
\end{multline}
\begin{multline}
\label{Constraint2}
L_{f}h_{i}(x) + L_{g}h_{i}(x)u(x) + \gamma sign({h_{i}(x)}) |h_{i}(x)|^{\rho} \geq 0 \\
\text{ $\forall$ i $\in \left\{ q'+1, \dots , q \right\}$}
\end{multline}
then under the feedback controller $u(x)$, for all initial conditions $x_{0} \in \mathcal{D}$, there exists $0 < T < \infty$ such that $x(T) \in \Gamma$.
\end{theorem}
\begin{proof}
By contradiction, suppose for some $x_0 \in \mathcal{D} \backslash \Gamma$ the control law $u(x)$ that satisfies \eqref{Constraint1} and \eqref{Constraint2} is such that there does not exist a finite time $0 < T < \infty$ so that $x(T) \in \Gamma$. In particular, then for all $t > 0$,  $min \bigg\{ h_{1}(x(t)), h_{2}(x(t)),\dots, h_{q}(x(t)) \bigg\} < 0$, where $x(t)$ is the solution to \eqref{ControlAffineSystem} initialized at $x(0)$ under the control law $u(x)$. By \eqref{Constraint2} for all $t > T_i  = \frac{|h_i(x_0)|^{1 - \rho}}{\gamma (1 - \rho)}$, we have $h_{i}(x(t)) \geq 0$ for all $i = \{ q'+1,\dots,q\}$ by Proposition 1. To that end, if we define $T' = \max\limits_{i = q'+1,\dots,q} \big\{ T_i \big\}$, then for all $t > T'$ we have, $min \bigg\{ h_{1}(x(t)), h_{2}(x(t)),\dots, h_{q'}(x(t)) \bigg\} < 0$. In particular, observe that
\begin{align}
\label{const}
\frac{d}{dt} \sum\limits_{i = 1}^{q'}\bigg\{ \alpha_{i} h_{i}(x(t)) \bigg\} = \sum\limits_{i = 1}^{q'} \bigg\{ \alpha_{i}(L_{f}h_{i}(x) + L_{g}h_{i}(x)u(x)) \bigg\}
\end{align}
so that by integration of \eqref{const} using the fundamental theorem of calculus and \eqref{Constraint1}, we have
\begin{align*}
\sum\limits_{i = 1}^{q'}\bigg\{ \alpha_{i} h_{i}(x(t)) \bigg\} \geq \gamma (t-T') + \sum\limits_{i = 1}^{q'}\bigg\{ \alpha_{i} h_{i}(x(T')) \bigg\}
\end{align*}
We observe that as $t \rightarrow \infty$, $\sum\limits_{i = 1}^{q'} \bigg\{ \alpha_i h_{i}(x(t))\bigg\} \rightarrow \infty$. But this is a contradiction since $h_{i}(x(t))$ for $i = \{ 1,2\dots,q'\}$ is bounded i.e. $\sum\limits_{i = 1}^{q'} \bigg\{ \alpha_i h_{i}(x(t))\bigg\} < \sum\limits_{i = 1}^{q'} \alpha_i M_i$. This proves that there exists a $0 < T < \infty$ such that $x(T) \in \bigcap\limits_{i = 1}^{q'} \left\{ x \in \mathbb{R}^{n} | h_{i}(x) \geq 0 \right\}$.
\end{proof}
\newtheorem{cor}{\textbf{Corollary}}

\color{black}
We remark that \cite{c15} proposes a more restrictive solution to the constrained reachability problem with desired level sets being individually defined by multiple functions in a QP. In particular,  \cite{c15} allows for the set of control laws $\underline{\mathcal{U}}(x)$ given by
\begin{multline}
\underline{\mathcal{U}}(x) = \bigg\{ u \in \mathbb{R}^{m} \bigg| L_{f}h_{i}(x) + L_{g}h_{i}(x)u(x) + \\
\gamma \cdot sign({h_{i}(x)}) \cdot |h_{i}(x)|^{\rho} \geq 0 \\
\text{ $\forall$ i $\in \left\{ 1, \dots , q \right\}$} \bigg\} \text{,}
\end{multline}
Note that this is equivalent to taking $q' = 0$ in Theorem 1. To that end define,
\begin{multline}
\mathcal{U}(x) = \bigg\{ u \in \mathbb{R}^{m} \bigg| \text{\eqref{Constraint1} and \eqref{Constraint2} are satisfied} \bigg\}
\end{multline}
then we can formulate the following corollary

\begin{cor}
The set $\mathcal{U}(x)$ is a superset to the set $\underline{\mathcal{U}}(x)$ i.e. $\mathcal{U}(x) \supset \underline{\mathcal{U}}(x)$. \QEDB
\end{cor}

From \cite{c15}, finite time barrier certificates also possess the property of forward invariance. This allows us to encode the global constraints as well as other additional system constraints, discussed in Section III, as invariance conditions in a QP. To that end, we formulate the following remarks.

\newtheorem{rmk}{\textbf{Remark}}
\begin{rmk}
In addition, suppose we require that $x(t) \in \Sigma$ as in \eqref{avoidprob} for all $t > 0$, then we add these additional constraints individually in the QP as constraints of the form \eqref{Constraint2} of Theorem 1. In order to solve the constrained reachability problem of the form $R(\Sigma, \Gamma)$ as in Definition 3 where $\Gamma$ and $\Sigma$ are of the form \eqref{reachprob} and \eqref{avoidprob} respectively, we can use Theorem 1 for the reachability problem along with constraints of the form \eqref{Constraint2}  which ensure invariance as discussed in Proposition 1, for the safety and avoid problem.
\end{rmk}

\begin{rmk}{\footnote{Online quadratic programs are solved easily using non strict inequalities, so we introduce parameter $\epsilon>0$ where we take $\epsilon$ to be small.}}
Note that if $\overline{\mathcal{G}_{\pi_{i}^{r}}}$ is a part of the definition of $\Sigma$ as in \eqref{avoidprob}, then we use $\overline{\mathcal{G}_{\pi_{i}^{r}}} = \overline{\mathcal{G}_{i}^{r}} = \left\{ x \in \mathcal{D}^{N} | h_{r}(x_i) < 0 \right\}$ for all $r \in \mathcal{R}$, $i \in \mathcal{I}$. For each $\pi_i^r$ with $i\in\mathcal{I}$ and $r\in\mathcal{R}$, let $\bar{h}_{r,\epsilon}(x_i)=-h_r(x_i)-\epsilon$ so that we have $\{x \in \mathcal{D}^{N}|\bar{h}_{r,\epsilon}(x_i)\geq 0\}\subset \overline{\mathcal{G}_i^r}$. Likewise, for $c\in\mathcal{C}$, let $\bar{h}^\text{globe}_{c,\epsilon}(x)=-\hglobe_c(x)-\epsilon$ so that $\{x \in \mathcal{D}^{N}|\bar{h}^\text{globe}_{c,\epsilon}(x)\geq 0\}\subseteq \overline{\Gglobe_c}$. The constraints can then be encoded in the QP as $L_{f}\bar{h}_{r,\epsilon}(x_i) + L_{g}\bar{h}_{r,\epsilon}(x_i)u(x) + \gamma \cdot sign(\bar{h}_{r,\epsilon}(x_i)) \cdot |\bar{h}_{r,\epsilon}(x_i)|^{\rho} \geq 0$ and $L_{f}\bar{h}^\text{globe}_{c,\epsilon}(x) + L_{g}\bar{h}^\text{globe}_{c,\epsilon}(x)u(x) + \gamma \cdot sign(\bar{h}^\text{globe}_{c,\epsilon}(x)) \cdot |\bar{h}^\text{globe}_{c,\epsilon}(x)|^{\rho} \geq 0$, for all $i \in \mathcal{I}$, $r \in \mathcal{R}$ and $c \in \mathcal{C}$ as described in \eqref{Constraint2}.
\end{rmk}

\begin{figure}[thpb]
    \includegraphics[width=8.7cm,height=7cm,scale = 1]{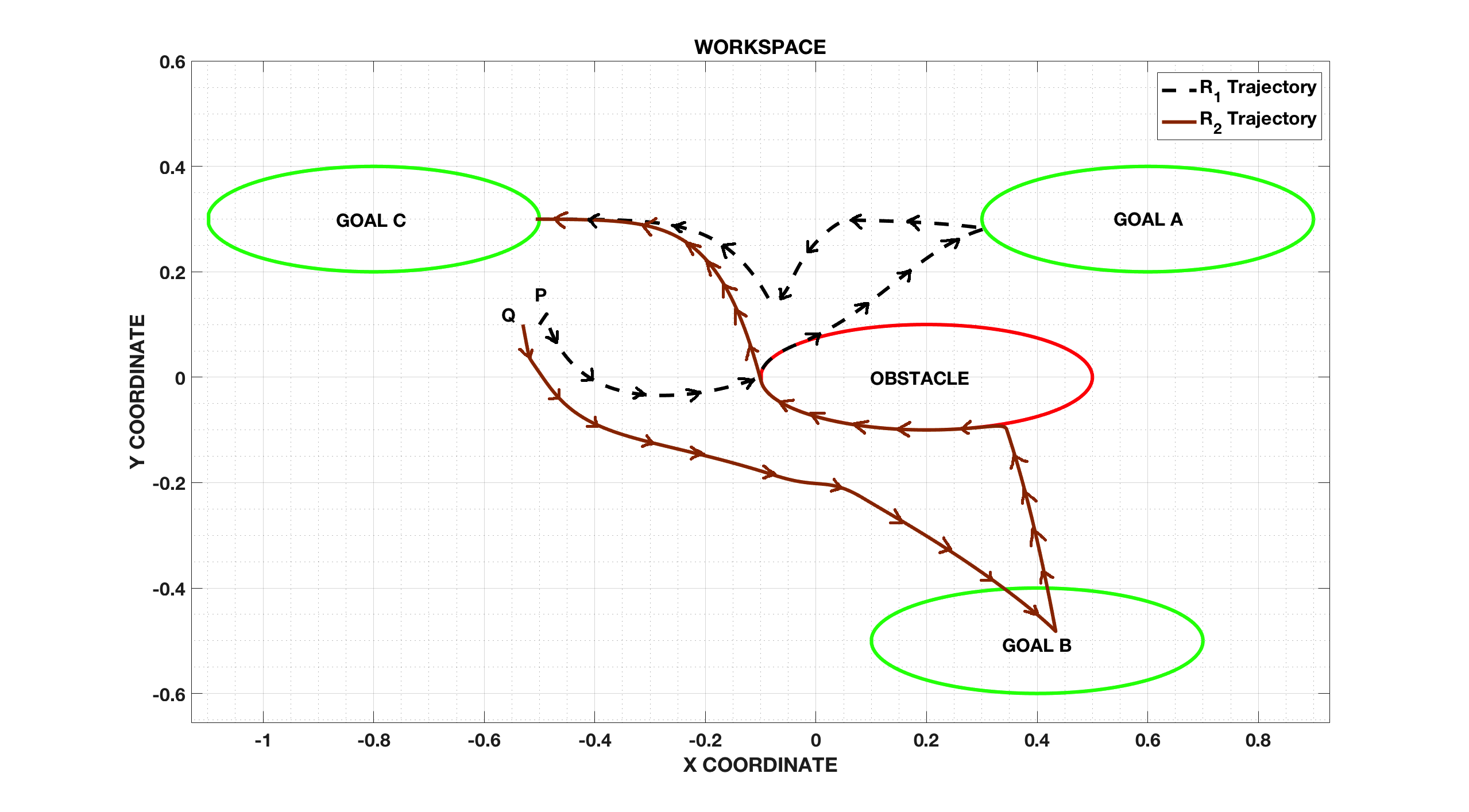}
      \caption{A simulated trajectory for $\mathtt{R}_{1}$ and $\mathtt{R}_{2}$ which satisfies the task specification, ``$\mathtt{R}_{1}$ should go to goal A and $\mathtt{R}_{2}$ should go to goal B, and then both robots should go to goal C, all the while ensuring that the safety constraint is satisfied and the obstacle is avoided".}
      \label{figurelabel}
\end{figure}
\section{ILLUSTRATIVE EXAMPLE}
Consider a two robot (single integrator dynamics) homogeneous multi-agent system in the domain $\mathcal{D}^{2} \subset \mathbb{R}^{4}$ with three regions $\mathcal{R} = \left\{ A,B,C, O\right\}$ and agents indexed by the set $\mathcal{I} = \left\{ 1,2\right\}$. This is as shown in Fig. 1. The state of the system is $x \in \mathbb{R}^{4}$. We first require robot $\mathtt{R}_{1}$ to visit goal A and robot $\mathtt{R}_{2}$ to visit goal B. Then, both the robots must visit goal C. This process needs to be repeated infinitely often. However, there is a caveat to this task specification. We enforce a connectivity constraint between the two robots which is a function of the state of $\mathtt{R}_{2}$. In addition to this, the robots must always avoid the obstacle region `O'. It is important to note that we do not require the agents to enter the goal regions simultaneously. However, after a certain time both the agents must be inside their respective goal regions. The function for the level set of each region of interest is
\begin{align}
h_{r}(x_i) = 1 - (x_i - C_{r})^{T}P_{r}(x_i - C_{r}) \text{ , $ \forall r \in \mathcal{R}$, $\forall i \in \mathcal{I}$.}
\end{align}
Here, $P_{r}$ is a positive definite matrix, $C_{r}$ is the center of the region of interest and $x_i$ is the state of agent $i$ of the system.
\begin{figure}[thpb]
      \centering
    \includegraphics[width=8.5cm,height=6.5cm,scale = 1]{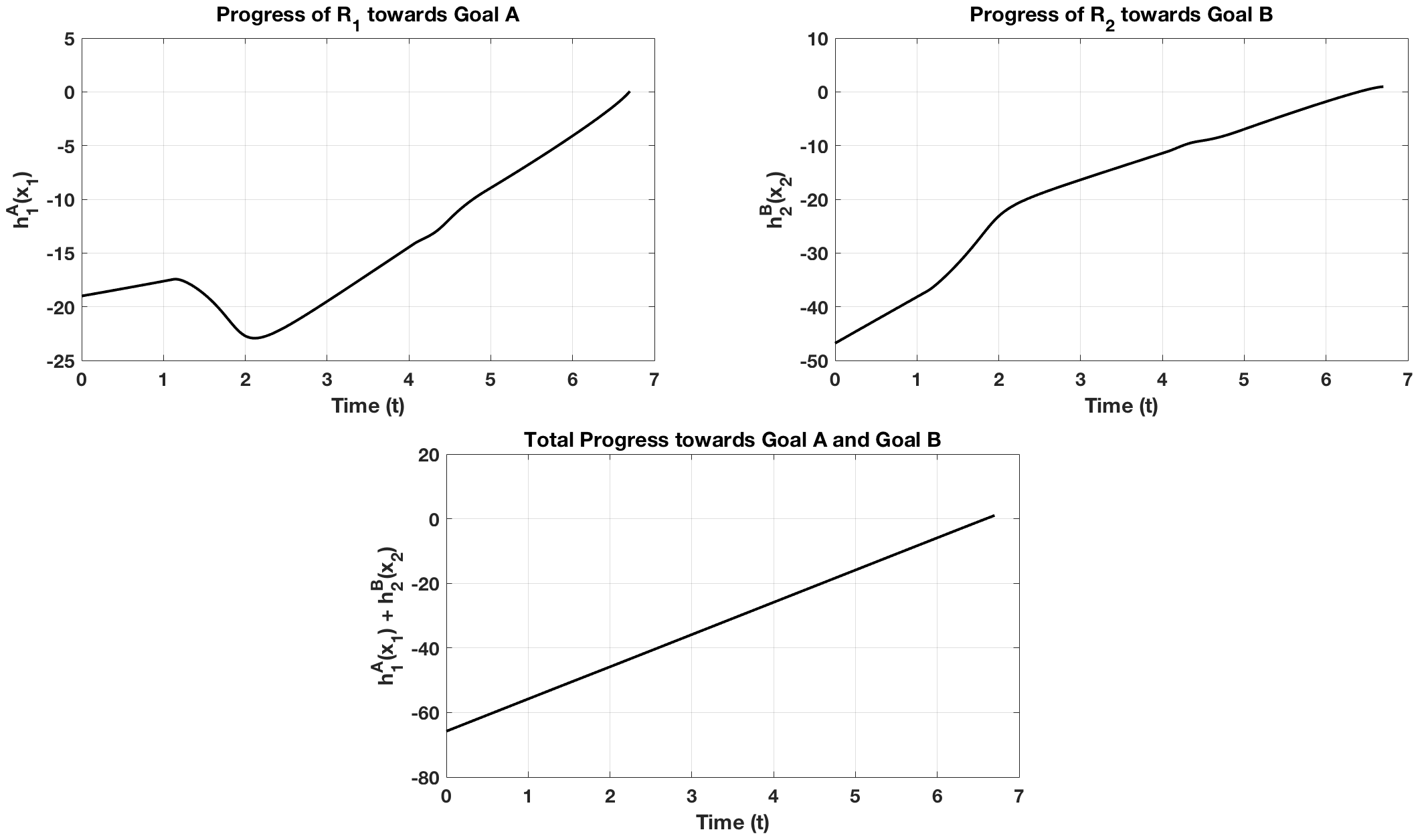}
      \caption{Level sets for goal A and goal B, along with the net progress and rate of progress towards both goals. Even though $\mathtt{R}_{1}$ moves away from goal A for a brief moment, the net total progress towards the goals is increasing at all times.}
      \label{figurelabel}
      \end{figure}
The three goal regions A, B and C are defined as \eqref{Goal}, and the obstacle `O' is defined as \eqref{Obstacle}.
\color{black}
The additional connectivity constraint is given as
\begin{equation}
h_{globe}(x) = d_{globe}^{2}(x) - ||x_{2} - x_{1}||^{2},
\end{equation}
where $d_{globe} : \mathcal{D}^{2} \rightarrow \mathbb{R}$ is the connectivity distance between the two agents that needs to be maintained, and $|| x_{2} - x_{1}||$ is the inter-agent distance. We consider
\begin{equation}
d_{globe}^{2}(x)  = (x_{2,1} + \delta_{1})^{2} + \delta_{2},
\end{equation}
where $\delta_{1}$ and $\delta_{2}$ are constants, and $x_{2,1}$ is the $x$ coordinate of $\mathtt{R}_{2}$ in the workspace. Such a constraint captures a situation in which the robots have poor connectivity in certain areas of the workspace, which requires them to maintain a closer distance with each other. In areas where the robots have strong connectivity, they are free to maintain a larger distance from each other. This constraint is equivalent to the global constraints $h_{c}^{globe}(x)$ discussed in section III.

The global LTL specification to be satisfied by the multi-agent system is
\begin{align}
\phi = \Box (\Diamond ((\pi_{1}^{A} \wedge \pi_{2}^{B}) \wedge \Diamond \big( \pi_{1}^{C} \wedge \pi_{2}^{C} \big) ) ) \wedge \Box (\pi^{globe} \wedge \neg \pi_{1}^{O} \wedge \neg \pi_{2}^{O} ) \text{.}
\end{align} 
\color{black}
The lasso-type constrained reachability sequence for this example from Definition 5 is $\mathcal{R}_{lasso} = (R_1 R_2)^{\omega}$.

From \eqref{AP}, $\Pi = \{ \pi_{1}^{A}, \pi_{2}^{B}, \pi_{1}^{C}, \pi_{2}^{C}, \pi^{globe}, \pi_{1}^{O}, \pi_{2}^{O} \}$. We assume that the robots start off within the connectivity radius $d_{globe}$. Applying the formalism from Definition 3 and Definition 4, we have $a_{1}^{\top} = \{ \pi^{globe}\}$, $a_{1}^{\bot} = \{ \pi_{1}^{A}, \pi_{2}^{B}, \pi_{1}^{C}, \pi_{2}^{C}, \pi_{1}^{O}, \pi_{2}^{O}\}$, $a_{2}^{\top} = \{ \pi_{1}^{A}, \pi_{2}^{B}, \pi^{globe}\}$, and $a_{2}^{\bot} = \{ \pi_{1}^{O}, \pi_{2}^{O}\}$. Hence, the reachability set from \eqref{reachprob} is given by $\Gamma_1 = \bigg( \mathcal{G}_{\pi_{1}^{A}} \cap \mathcal{G}_{\pi_{2}^{B}} \bigg)$.  The safety set from \eqref{avoidprob} is given by, $\Sigma_1 = \mathcal{G}_{\pi^{globe}} \cap \bigg( \overline{\mathcal{G}_{\pi_{1}^{O}}} \cap \overline{\mathcal{G}_{\pi_{2}^{O}}} \bigg)$. Thus, the constrained reachability problem to be solved is $R_1(\Sigma_1, \Gamma_1)$.

The constraints encoded in the QP are,
\begin{align*}
& \frac {\partial (h_{A}(x_{1}) + h_{B}(x_{2}))}{\partial x}u \geq - \gamma \cdot sign(min \left\{ h_{A}(x_{1}), h_{B}(x_{2})\right\}) \\
& \frac {\partial h_{O}(x_i)}{\partial x}u \geq - \gamma \cdot sign(h_{O}(x_i)) \cdot |h_{O}(x_i)|^{\rho} \text{, for all $i \in \mathcal{I}$}\\
& \frac {\partial h_{globe}(x)}{\partial x}u \geq - \gamma \cdot sign(h_{globe}(x)) \cdot |h_{globe}(x)|^{\rho}
\end{align*}

For the second reachability problem, from Definition 3 and Definition 4, we have $a_{1}^{\top} = \{\pi_{1}^{A}, \pi_{2}^{B}, \pi^{globe}\}$, $a_{1}^{\bot} = \{ \pi_{1}^{C}, \pi_{2}^{C}, \pi_{1}^{O}, \pi_{2}^{O}\}$, $a_{2}^{\top} = \{ \pi_{1}^{C}, \pi_{2}^{C}, \pi^{globe}\}$, and $a_{2}^{\bot} = \{ \pi_{1}^{O}, \pi_{2}^{O}\}$. The corresponding reachability problem and safety problem are $\Gamma_2 = \bigg( \mathcal{G}_{\pi_{1}^{C}} \cap \mathcal{G}_{\pi_{2}^{C}} \bigg)$, and $\Sigma_2 = \mathcal{G}_{\pi^{globe}} \cap \bigg( \overline{\mathcal{G}_{\pi_{1}^{O}}} \cap \overline{\mathcal{G}_{\pi_{2}^{O}}} \bigg)$
Thus, the constrained reachability problem to be solved is $R_2(\Sigma_2, \Gamma_2)$.

Similarly, the constraints encoded in the QP are,
\begin{align*}
& \frac {\partial (h_{C}(x_{1}) + h_{C}(x_{2}))}{\partial x}u \geq - \gamma \cdot sign(min \left\{ h_{C}(x_{1}), h_{C}(x_{2})\right\}) \\
& \frac {\partial h_{O}(x_i)}{\partial x}u \geq - \gamma \cdot sign(h_{O}(x_i)) \cdot |h_{O}(x_i)|^{\rho} \text{, for all $i \in \mathcal{I}$}\\
& \frac {\partial h_{globe}(x)}{\partial x}u \geq - \gamma \cdot sign(h_{globe}(x)) \cdot |h_{globe}(x)|^{\rho}
\end{align*}

By solving these two QPs, we solve the lasso-type constrained reachability sequence, $\mathcal{R}_{lasso} = ( R_1 R_2 )^{\omega}$ as per \eqref{reachlasso}.

\color{black}
\section{SIMULATION AND EXPERIMENTAL RESULTS}
The generated trajectory in MATLAB is shown in Fig. 1. Fig. 2 illustrates how our theorem is effective by ensuring that the system makes total positive (increasing) progress towards its goals for all $t \geq 0$, even though an individual robot moves away from its goal momentarily ($\mathtt{R}_{1}$ moves away from goal A in the interval from P to Q in Fig. 1. This corresponds to the dip in the level set $h_{1}^{A}(x_{1})$ in the first plot in Fig. 2). It is important to note that the connectivity constraint is maintained for all time of the simulation and experiment. We also execute our algorithm on the Robotarium testbed \cite{c16} and provide a video of the experiment (https://youtu.be/Gnga3k2BHWg). The trajectories followed by the robots is consistent with the trajectory in Fig. 1.

\section{CONCLUDING REMARKS}
In this paper we provided a theoretical framework to synthesize controllers for continuous time multi-agent systems, given a linear temporal logic task specification using finite time control barrier certificates. We formulated a theorem on the composition of multiple bounded finite time barrier certificates. The proposed framework results in a larger set of feasible control laws as compared to methods such as \cite{c15}. By solving a sequence of constrained reachability problems by means of quadratic programs, we solve a lasso-type constrained reachability sequence that synthesizes system trajectories whose traces satisfy the given LTL specification.
\addtolength{\textheight}{-12cm}


\begin{thebibliography}{99}

\bibitem{c1} E. M. Wolff, U. Topcu and R. M. Murray, ``Automaton-guided controller synthesis for nonlinear systems with temporal logic," 2013 IEEE/RSJ International Conference on Intelligent Robots and Systems, Tokyo, 2013, pp. 4332-4339.
\bibitem{c2} TT. Wongpiromsarn, U. Topcu and A. Lamperski, ``Automata Theory Meets Barrier Certificates: Temporal Logic Verification of Nonlinear Systems," in IEEE Transactions on Automatic Control, vol. 61, no. 11, pp. 3344-3355, Nov. 2016.
\bibitem{c3} C. Baier and J.-P. Katoen, Principles of Model Checking, MIT Press, 2008.
\bibitem{c4} L. Wang, A. D. Ames and M. Egerstedt, ``Safety Barrier Certificates for Collisions-Free Multirobot Systems," in IEEE Transactions on Robotics, vol. 33, no. 3, pp. 661-674, June 2017.
\bibitem{c5} L. Wang, A. Ames and M. Egerstedt, ``Safety barrier certificates for heterogeneous multi-robot systems," 2016 American Control Conference (ACC), Boston, MA, 2016, pp. 5213-5218.
\bibitem{c6} A. D. Ames, J. W. Grizzle and P. Tabuada, ``Control barrier function based quadratic programs with application to adaptive cruise control," 53rd IEEE Conference on Decision and Control, Los Angeles, CA, 2014, pp. 6271-6278.
\bibitem{c7} A. D. Ames, X. Xu, J. W. Grizzle and P. Tabuada, ``Control Barrier Function Based Quadratic Programs for Safety Critical Systems," in IEEE Transactions on Automatic Control, vol. 62, no. 8, pp. 3861-3876, Aug. 2017.
\bibitem{c8} G. E. Fainekos, H. Kress-Gazit and G. J. Pappas, ``Temporal Logic Motion Planning for Mobile Robots," Proceedings of the 2005 IEEE International Conference on Robotics and Automation, 2005, pp. 2020-2025.
\bibitem{c9} S. Prajna, A. Jadbabaie and G. J. Pappas, ``A Framework for Worst-Case and Stochastic Safety Verification Using Barrier Certificates," in IEEE Transactions on Automatic Control, vol. 52, no. 8, pp. 1415-1428, Aug. 2007.
\bibitem{c10} S. Prajna, and A. Jadbabaie, ``Safety Verification of Hybrid Systems Using Barrier Certificates", International Workshop on Hybrid Systems: Computation and Control (HSCC 2004), pp 477-492
\bibitem{c11} A. Ulusoy, S. L. Smith, X. C. Ding and C. Belta, ``Robust multi-robot optimal path planning with temporal logic constraints," 2012 IEEE International Conference on Robotics and Automation, Saint Paul, MN, 2012, pp. 4693-4698.
\bibitem{c12} Y. Kantaros and M. M. Zavlanos, ``Sampling-Based Control Synthesis for Multi-robot Systems under Global Temporal Specifications," 2017 ACM/IEEE 8th International Conference on Cyber-Physical Systems (ICCPS), Pittsburgh, PA, 2017, pp. 3-14.
\bibitem{c13} M. Rauscher, M. Kimmel and S. Hirche, ``Constrained robot control using control barrier functions," 2016 IEEE/RSJ International Conference on Intelligent Robots and Systems (IROS), Daejeon, 2016, pp. 279-285.
\bibitem{c14} M. Guo, K. H. Johansson and D. V. Dimarogonas, ``Motion and action planning under LTL specifications using navigation functions and action description language," 2013 IEEE/RSJ International Conference on Intelligent Robots and Systems, Tokyo, 2013, pp. 240-245.
\bibitem{c15}A. Li, L. Wang, P. Pierpaoli, and M. Egerstedt, ``Formally Correct Composition of Coordinated Behaviors using Control Barrier Certificates", 2018 IEEE/RSJ International Conference on Intelligent Robots and Systems
\bibitem{c16} D. Pickem, P. Glotfelter, L. Wang, M. Mote, A. Ames, E. Feron, and M. Egerstedt, ``The Robotarium: A remotely accessible swarm robotics research testbed," 2017 IEEE International Conference on Robotics and Automation (ICRA), Singapore, 2017, pp. 1699-1706.
\bibitem{c17}M. Kloetzer, X. C. Ding and C. Belta, ``Multi-robot deployment from LTL specifications with reduced communication," 2011 50th IEEE Conference on Decision and Control and European Control Conference, Orlando, FL, 2011, pp. 4867-4872.
\bibitem{c18}E. A. Gol and C. Belta, ``Time-constrained temporal logic control of multi-affine systems", Nonlinear Analysis: Hybrid Systems (NAHS), 2013
\bibitem{c19} M. Tobenkin, I. Manchester, and R. Tedrake, ``Invariant funnels around trajectories using sum-of-squares programming", https://arxiv.org/abs/1010.3013
\color{black}
\bibitem{c20} S. P. Bhat and D. S. Bernstein, ``Finite-time stability of continuous autonomous systems", SIAM Journal on Control and Optimization, 38(3):751?766, 2000


\end{thebibliography}
\end{document}